\documentclass[12pt,oneside]{book}
\usepackage{amsmath,amssymb}
\usepackage{amsthm} 

\theoremstyle{plain}
\newtheorem{theorem}{Theorem}
\theoremstyle{definition}
\newtheorem{definition}{Definition} 

\usepackage{graphicx}   
\usepackage{stackengine} 
\usepackage{scalerel}    
\stackMath

\long\def\comment#1{}

\newcommand{\be}{\begin{equation}}
\newcommand{\ee}{\end{equation}}
 
\newtheorem{Thm}{Theorem}          
\newtheorem{corollary}[Thm]{Corollary}   

\newtheorem{remark}{Remark} 
\newtheorem{prop}{Proposition}
 
\newtheorem{assumption}{Assumption}

\newfont{\bbb}{msbm10 scaled 700}

\newfont{\bb}{msbm10 scaled 1100}

\newcommand{\xv}{{\bf x}}

\newcommand{\RNum}[1]{\uppercase\expandafter{\romannumeral #1\relax}}

\newcommand{\eqdef}{\stackrel{\Delta}{=}}

\newcommand{\<}{\left\langle}
\renewcommand{\>}{\right\rangle}

\usepackage{dsfont}
\usepackage{amsmath}
\DeclareMathOperator*{\argmin}{arg\,min}

\usepackage[margin=1in]{geometry}
\usepackage{hyperref}
\usepackage{circuitikz}
\usepackage{booktabs}
\usepackage{amssymb}

\usepackage{graphicx}
\usepackage{epstopdf}
\usepackage{multirow}
\usepackage{cite}

\title{\textbf{Technical Report for Dissipativity Learning in Reproducing Kernel Hilbert Space}}
\author{Xiuzhen Ye, Wentao Tang}
\date{\today}

\begin{document}
\maketitle   
This work presents a nonparametric framework for \emph{dissipativity learning} in reproducing kernel Hilbert spaces (RKHS), which enables data-driven certification of stability and performance properties for unknown nonlinear systems without requiring an explicit dynamic model. Dissipativity is a fundamental system property that generalizes Lyapunov stability, passivity, and finite $\mathcal{L}_2$-gain conditions through an energy-balance inequality between a storage function and a supply rate. Unlike prior parametric formulations that approximate these functions using quadratic forms with fixed matrices, the proposed method represents them as Hilbert–Schmidt operators acting on canonical kernel features, thereby capturing nonlinearities implicitly while preserving convexity and analytic tractability. The resulting operator optimization problem is formulated in the form of a one-class support vector machine (OC-SVM) and reduced, via the representer theorem, to a finite-dimensional convex program expressed through kernel Gram matrices. Furthermore, statistical learning theory is applied to establish generalization guarantees, including confidence bounds on the dissipation rate and the $\mathcal{L}_2$-gain. Numerical results demonstrate that the proposed RKHS-based dissipativity learning method effectively identifies nonlinear dissipative behavior directly from input–output data, providing a powerful and interpretable framework for model-free control analysis and synthesis.

\section{Introduction}\label{sec_introduction}
Data-driven techniques, where statistical and machine learning methods are utilized for modelling and extracting system properties, are the main trending of the research on control systems. This stems from the fact that the potential of easing the derivation and maintenance of dynamics models~\cite{H_IS_13}. In model-based data-driven approaches, research effort has been dedicated to the approaches of state-space model construction~\cite{L_SystemIdentification_98}, parameter estimation~\cite{KK_JohgWiley_95}, Koopman operators~\cite{strasser_TAC_24}~\cite{Kostic_ANIPS_22}, or neural networks~\cite{GH_SpringerScience}, etc. In these approaches, challenges persist in developing models that accurately capture the underlying system dynamics, which consequently leaves the subsequent control design performance unguaranteed. 
In general, the state-of-the-art model-based approaches still involve model identification scheme and are not truly model independent.

With this consideration, instead of learning system dynamics, we are motivated to develop a {\it model free} approach that seeks for essential control-relevant property. The rationale of learning system properties lies in its simplicity and sophistication compared to learning a full dynamic model, while maintaining a more direct connection to the resulting control performance. Specifically, the model-free system property learning approach that this work focuses on is {\it dissipativity learning}, proposed in the author's earlier works~\cite{TD_ACC_19}~\cite{TD_CCE_19}~\cite{TD_SCL_21}.
Dissipativity is a fundamental property of nonlinear dynamics~\cite{LB_Dissipative_13}, describing system-wide behavior in the form of a Lyapunov-like inequality. Particularly, it refers to the existence of a state-dependent {\it storage function} whose accumulation is upper-bounded by a {\it supply rate} that depends on the system’s inputs and outputs~\cite{ZJ_Springer_20}. Since transport phenomena are governed by the second law of thermodynamics, dissipativity is known to be a common characteristic of process systems~\cite{BL_Springer_07}.  

In the author’s previous works~\cite{TD_CCE_19}~\cite{TD_SCL_21}, dissipativity learning control (DLC) was proposed as a data-driven control framework, where the dissipativity property of unknown nonlinear dynamics is learned from input-output trajectories. In these works, the supply rate to be estimated from data is parameterized as a quadratic form, so that the matrix defining the quadratic form is to be learned. Naturally, quadratic forms are suitable for systems that are linear or approximately linear, which restricts their applicability to processes that exhibit severe nonlinearity. In addition, the representation of system-wide behavior by the collected trajectories often rely on restrictive assumptions and the complexity of sampling input-output trajectories is prohibitive. 

Therefore, in this work, we formulate dissipativity learning in nonparametric framework in reproducing kernel Hilbert spaces (RKHS)~\cite{BT_Springer_11}~\cite{PR_Cambridge_16}. The information of inputs, outputs, and states is expressed in terms of their canonical features corresponding to the kernel functions in RKHS via kernel feature embeddings~\cite{GG_arXiv_21}. Such canonical features maps data samples nonlinearly into corresponding RKHSs, thereby accounting for the nonlinearity of the dynamics implicitly. Hence, nonlinear functions in the original state-space are represented by the Hilbert-Schmidt (H-S) operators acting on kernel feature maps. This is a generalization of the previous works~\cite{TD_CCE_19}~\cite{TD_SCL_21} that adopt quadratic forms of dissipativity and the matrices parameterizing these quadratic forms are learned. With kernel feature embeddings, the functions remains quadratic in the features to preserve convexity and analytic tractability, while nonlinearity is captured for the dissipative behavior. 

The dissipativity learning problem is hence formulated as an operator convex optimization problem with constraints that enforce dissipativity and stability properties (e.g., the $L_2$ gain). Formally, by extending inner product and norm concepts from matrices to operators, the learning of H-S operators preserves the formulation from finite-dimensional spaces and has the form of one-class support vector machine (OC-SVM)~\cite{OCSVM_ACM_13}. Due to the representation theorem, our formulation is reducible to a tractable, finite-dimensional problem, despite the infinite-dimensional nature. Therefore, the proposed approach enables flexible learning of dissipative properties in general nonlinear forms, without requiring explicit knowledge of the system dynamics. 

For a theoretical analysis, we adopt statistical learning theory to establish bounds on generalized performance, namely a confidence lower bound on the dissipation rate and a confidence upper bound on the guaranteed $L_2$ gain. This work exemplifies the use of operators and kernel tricks in RKHS as generic tools for learning the system behaviors.  

The remainder of this paper is organized as follows. In Section~\ref{sec_DissipativitySystems}, we first introduce the preliminaries of dissipative systems theory. Then in Section~\ref{sec_parameterlearning}, data-driven modelling and estimation of dissipativity property via OC-SVM are disscused. Then, in Section~\ref{sec_RKHS}, we derived the optimization formulations for learning in RKHS that is solved by OC-SVM, followed with generalized error analysis in Section~\ref{sec_ErrorAnalysis}. The proposed method is examined with application to systems with complex nonlinear dynamics in Section~\ref{sec_simulation}. Finally conclusions are given in Section~\ref{sec_conclusion}.

{\it Notations}. Upper case letters are used to represent matrices and linear operators. Lower case letters are for scalars and column vectors. For a matrix $A \in \mathbb{R}^{n \times n}$ whose entries are written as $A_{i,j}, i, j \in \{1,2,\cdots,n\}$, its trace and Frobenius norm are denoted by $\textnormal{tr}(A) = \sum_{i=1}^n A_{i,i}$ and $\|A\|_F = \left(\sum_{i, j = 1}^n |A_{i,j}|^2\right)^{1/2}$, respectively. We use $\langle A, B\rangle \eqdef \textnormal{tr}(A^{\sf T} B)$ to denote the inner product between two matrices $A, B \in \mathbb{R}^{n \times n}$, where $A^{\sf T}$ is the transpose of the matrix $A$. Identity matrix is denoted by $I$. Calligraphic letters are used for sets (e.g., ${\cal X}$) or Hilbert Space (e.g., ${\cal H}$). Moreover, a Reproducing Kernel Hilbert Space (RKHS) generated by kernel $K$ is denoted by ${\cal H}_K$.
Capitalized Greek letters (e.g., $\Pi$) are used to denote operators, and its Hilbert-Schmidt norm is $\| \Pi \|_{\textnormal{HS}}$. We use also $a \otimes b$ to denote Kronecker product between $a$ and $b$. Fourier transform of a function $\cdot $ is ${\cal F}\{\cdot\}$ and its inverse Fourier transform is ${\cal F}^{-1}\{\cdot\}$.

\section{Dissipative Systems}\label{sec_DissipativitySystems}
Consider an unknown system governed by discrete time nonlinear dynamics 
\begin{equation}\label{eq_system}
\begin{aligned}
\begin{cases}
&x_{t+1} = f(x_{t}, u_{t}) \\
&y_{t} = h(x_{t}, u_{t}),
\end{cases}
\end{aligned}
\end{equation}
where $x_{t} \in \mathbb{R}^{n_x}$, $u_{t} \in \mathbb{R}^{n_u}$, and $y_{t} \in \mathbb{R}^{n_y}$ denote the state, input and the output variables of the system at time $t \geq 0$, respectively.   
The maps $f: {\cal X} \times {\cal U} \to \mathbb{R}^{n_x}$ is the drift and 
$h: {\cal X} \times {\cal U} \rightarrow \mathbb{R}^{n_y} $ is the output function. It is assumed that $\cal X$ is assumed to be compact and $f, h$ are locally Lipschitz in $(x,u) \in {\cal X} \times {\cal U}$. We also assume $f(0) = 0$, that is, the origin is an equilibrium point of~\eqref{eq_system}. 
For an initial condition $x_0 \in \mathbb{R}^{n_x}$, we assume the existence and uniqueness of the solution of~\eqref{eq_system}, and denote the solution at time $t$ by $x(t; x_0)$. 
 
Through out this paper, the system dynamics $f$ is unknown, whereas, the data samples $x \in {\cal X} \subset \mathbb{R}^{n_x}, u \in {\cal U} \subset \mathbb{R}^{n_u}$, and $y \in { \cal Y} \subset \mathbb{R}^{n_y}$ are available. The goal is to learn the properties of the system. Therefore, the learning approach relies only on the data sample of the system. Specifically, we focus on the dissipativity property of a system. Dissipativity is a universal property of physical and engineering systems, capturing the balance between storage function $V: \mathbb{R}^{n_x} \to \mathbb{R}$ and external supply rate $s: \mathbb{R}^{n_y} \times \mathbb{R}^{n_u} \to \mathbb{R}$. 
From a control-theoretic perspective, dissipativity links directly to stability, performance, and controller synthesis.

Specifically, defined by Willems in~\cite{Def_diss}, the change of the state dependent function (or {\it storage function}) $V(x)$ can not exceed the accumulation of an input and output dependent function (or {\it supply rate}) $s(z)$. 
More precisely, the dissipativity property is stated in the following.
\begin{definition}
    A dynamical system of the form in~\eqref{eq_system} is said to be {\it dissipative} in the storage function $V(x)$ with respect to the supply rate $s(y, u)$ if 
    \be\label{eq_def_dissipativity}
        V(x_{t+1}) - V(x_{t}) \le s(y_{t}, u_{t}) 
    \ee
    holds for all input $u_{t}$ and state $x_{t}, 0 \leq t \leq T$.
\end{definition}

From the definition, dissipativity is a fundamental property that connects the internal energy balance of a system with its input--output behavior. 
If inequality~\eqref{eq_def_dissipativity} holds, the storage function $V(x)$ can be interpreted as a generalized energy of the system, 
while the supply rate $s(y,u)$ captures the net inflow of energy. 
This perspective provides a powerful framework for \emph{stability analysis}, since a nonincreasing $V(x)$ implies bounded state trajectories and can certify potential convergence to the equilibrium for controlled systems. 
Moreover, dissipativity generalizes passivity and small-gain conditions, thereby enabling \emph{controller synthesis} for performance objectives such as finite $L_2$-gain and robust regulation. 
For partially known or unknown systems, learning the dissipativity property offers a model-agnostic approach for certifying desirable closed-loop properties, and provides a framework for synthesizing controllers directly from data. 
If the supply rate has a quadratic form 
\be
s(y, u) = [y^{\sf T} \ u^{\sf T}] \begin{bmatrix}
    \Pi_{yy} & \Pi_{yu} \\
    \Pi_{yu}^{\sf T} & \Pi_{uu}
\end{bmatrix} \begin{bmatrix}
    y\\
    u
\end{bmatrix},
\ee
where $\Pi_{yy}, \Pi_{yu}, \Pi_{uu}$ are symmetric matrices, then the system is said to be $(\Pi_{yy}, \Pi_{yu}, \Pi_{uu})$-dissipative. This concept of dissipativity is closely related to input-output $L_2$ stability~\cite{Khalil_NonlinearSystems_02}~\cite{RB_JPC_08}. More generally, a tighter upper bound of $L_2$-gain is given in the following theorem~\cite{TD_ACC_19}.
\begin{theorem}
    For a $(\Pi_{yy}, \Pi_{yu}, \Pi_{uu})$-dissipative system, if there exists positive real numbers $\alpha$ and $\beta$, such that
    \be
    I + \alpha \Pi_{yy} \prec 0, \ \ \begin{bmatrix}
        I + \alpha \Pi_{yy} & \alpha \Pi_{yu} \\
        \alpha \Pi_{yu}^{\sf T} & \alpha \Pi_{uu} - \beta I
    \end{bmatrix} \preceq 0, 
    \ee
    then the system has an $L_2$-gain not greater than $\beta^{1/2}$.
\end{theorem}

The above result highlights how the dissipativity property inherently encodes the Lyapunov asymptotic stability of the system as proven in~\cite{HM_TAC_03}.
This connection motivates the learning of dissipativity directly from data such that the inequality holds approximately and generalizes to unseen trajectories, as it enables data-driven certification of robustness without requiring an explicit dynamic model. 

\section{Dissipativity parameters and its learning: quadratic case}\label{sec_parameterlearning}
\subsection{Dissipativity parameters}\label{sec_parameters}
Note that the supply rate function $s(y, u)$ in~\eqref{eq_def_dissipativity} is both input and output dependent. To simplify the analysis and learning of the function $s(y, u)$, we start with joining two variables into a single combined variable $z$ such that
\be\label{eq_def_z}
z \eqdef
\begin{bmatrix}
    y \\
    u
\end{bmatrix} \in \mathbb{R}^{n_y + n_x}.
\ee
We assume a quadratic form for supply function and a Lyapunov function candidate for the storage function in~\eqref{eq_def_dissipativity}, that is,
\be\label{eq_def_QandP}
s(z) = z^{\sf T} Q z, \ \ V(x) = x^{\sf T} P x,
\ee
where $0 \prec Q = Q^{\sf T} \in \mathbb{R}^{(n_y + n_u) \times (n_y + n_u)}$, 
$0 \prec P = P^{\sf T} \in \mathbb{R}^{x_n \times x_n}$.
Hence, the dissipative property in~\eqref{eq_def_dissipativity} is captured by
\be\label{eq_eq4}
z_{t}^{\sf T} Q z_{t} - (x_{t+1} -x_{t})^{\sf T} P (x_{t+1} -x_{t}) \geq 0.
\ee
\begin{remark}
 The system is passive if $n_u = n_y$, and $Q = \begin{bmatrix}
        M& S\\
        S^{\sf T}&R
    \end{bmatrix}$ with $M, R= 0$ and $S = I$.
\end{remark}
\begin{definition}\label{def_QandP}
    The matrix $Q$ and $P$ are called the dissipativity parameters that are to be estimated. 
\end{definition}

More generally, considering a {\it dynamic multiplier} $\Omega \in {\cal R} {\cal H}_{\infty}^{z_n \times (y_n + u_n)}: (y_k, u_k) \rightarrow z_k$ in~\cite{T_IQClearning_25} acting on inputs and outputs separately, it yields that
\be\label{eq_DynamicMultiplier}
\Tilde{z}_{t} \eqdef \Omega(q^{-1}) \begin{bmatrix}
    y_{t} \\
    u_{t}
\end{bmatrix} =  \begin{bmatrix}
     \Omega_y(q^{-1})&0 \\
    0&\Omega_u(q^{-1})
\end{bmatrix}\begin{bmatrix}
    y_{t} \\
    u_{t}
\end{bmatrix} = \begin{bmatrix}
    z_{y_{t}} \\
    z_{u_{t}}
\end{bmatrix}.
\ee
Under the general dynamic multiplier $\Omega$, the system in~\eqref{eq_system} is said to be \emph{dissipative} with respect to the quadratic supply rate $s(\Tilde{z}_k) = \Tilde{z}_k^{\sf T} \Tilde{Q} \Tilde{z}_k$, 
    if for any trajectory of the aggregated dynamics $(f,\Omega)$ starting from the origin ($x(t_0) = 0$) and for any $t > t_0$, we have
    $\Tilde{z}_{t}^{\sf T} \Tilde{Q}  \Tilde{z}_{t} - (x_{t+1} -x_{t})^{\sf T} P (x_{t+1} -x_{t}) \geq 0$.
\begin{remark}
    A trivial dynamic multiplier $\Omega = I$ yields a quadratic form of supply rate function 
    \be\label{eq_supply}
    s(y, u) = [y^{\sf T} u^{\sf T}] \ Q \begin{bmatrix}
        y\\
        u
    \end{bmatrix}.
    \ee
\end{remark}
From the above remark, for simplicity, we assume that $\Omega = I$ in this paper. The resulting dissipativity coincides with~\eqref{eq_eq4} where $z$ is defined in~\eqref{eq_def_z}.
The dissipativity parameters $Q$ and $P$ established in~\eqref{eq_eq4} are to be estimated. 

By exploiting the system model in~\eqref{eq_system}, the following theorem provides the linear matrices inequality (LMI) for model-based dissipative certification for the linearized system.
\begin{theorem} 
The linearized system around the equilibrium point of~\eqref{eq_system}, i.e., 
$x_{t+1}=A x_{t}+B u_{t}, y_{t}=C x_{t}$ is dissipative if there exist $\rho > 0$, $0 \prec P = P^{\sf T}, 0 \prec Q =Q^{\sf T}$ such that
\begin{equation}\label{eq_lmi_margin}
\begin{bmatrix}
I & 0\\
0 & C^{\sf T}
\end{bmatrix}
Q
\begin{bmatrix}
I & 0\\
0 & C
\end{bmatrix}
-
\begin{bmatrix}
B^{\sf T}P B   & B^{\sf T}P A \\
A^{\sf T}P B   & A^{\sf T}P A - P
\end{bmatrix}
\;\succeq\; \rho\,I 
\end{equation}
holds. 

\begin{proof}
The proof is obtained by applying $y_{t} = Cx_{t}$ and elementary operations to~\eqref{eq_eq4}, and enforce strong positive definiteness by introducing $\rho$.  
\end{proof}
\end{theorem}

\subsection{OC-SVM for dissipativity parameters}\label{sec_OCSVM_diss_parameter}
By linearizing and exploiting the system model in~\eqref{eq_system}, the dissipative certificate can be obtained according to~\eqref{eq_lmi_margin}. More realistically, in data-driven, or model-free approach, the dynamic modeling procedure can be skipped. Instead, system properties such as the discussed dissipativity are estimated from data~\cite{TD_ACC_19, martin2023guarantees, koch2021provably}. 
For the estimation of the dissipativity parameters $Q$ and $P$ as defined in~\eqref{eq_def_QandP}, we suppose that $m$ trajectories $\{y^i, u^i \}_{i = 1}^m$ and $\{x^i\}_{i=1}^m$ are sampled independently, from a distribution of random input excitations such that the magnitude of the input satisfies a Gaussian distribution with zero mean and each magnitude lasts for a fixed duration. Formally, we denote this distribution as a probability measure $\mathbb{P}$. The goal is therefore to determine valid matrices $Q$ and $P$ such that for all $i$ and $t$, that is, from~\eqref{eq_eq4}, the following inequality holds approximately: 
$(z^i_t)^{ \sf T} Q z_{t}^i - (x_{t+1}^i -x_{t}^i)^{\sf T} P (x^i_{t+1} -x^i_{t}) \geq 0$.
Hence, after elementary operations, the dissipativity inequality becomes
$ \langle Q,   z_{t} z_{t}^{\sf T} \rangle + \langle P,  x_{t}x_{t}^{\sf T} - x_{t+1}x_{t+1}^{\sf T} \rangle \geq 0$,
where $ \langle \cdot, \cdot \rangle$ is inner product.
\begin{definition}
    An input-output trajectory $\{ y_{t}, u_{t}\}_{t=1}^T$ that determines $z_{t} = [y_{t} \ u_{t}]^{\sf T}$, the corresponding dual dissipativity parameters refer to
    \be\label{eq_dualDissipativityParameter}
    \begin{cases}
        \Gamma_t \!\!\!\! & \eqdef z_{t}z_{t}^{\sf T}, \\
        \Delta_t \!\!\!\! & \eqdef x_{t}x_{t}^{\sf T} -  x_{t+1} x_{t+1}^{\sf T}, 
    \end{cases}
    \ee
\end{definition}
for all $t$.
Hence, with the dual dissipativity parameters of the sampled trajectories $\{\Gamma^{i}\}_{i = 1}^{n}$ and $\{ \Pi^{i}\}_{i = 1}^{n}$, we seek $Q$ and $P$ such that for all $t$, 
\be\label{eq_parametric_diss}
\biggl \langle Q, \Gamma_t^{i} \biggr \rangle - \biggl \langle P, \Delta_t^{i} \biggr \rangle \geq 0,
\ee
approximately. This problem is amenable to one-class support vector machine (OC-SVM), where we maximize the margin of the inequality, i.e., a nonnegative value $\rho >0$ such that $\langle Q, \Gamma_t^{i} \rangle - \langle P, \Delta_t^{i} \rangle \geq \rho $ for all $i$ and $t$, but penalizing the norm of the slope, i.e., $\|Q\|_F + \|P\|_F$. Equivalently, the problem is to minimize $\|Q\|_F^2 + \|P\|_F^2$ while rewarding $\rho$. To ensure a strong positive definiteness, we impose the constraint $P \succeq \epsilon I$ for some scalar $\epsilon > 0$. Similarly, we impose supply rate to be $\leq -\|y\|^2 + \alpha \|u\|^2$ for an ${\cal L}_2$-gain, that is $Q\prec  \begin{bmatrix}
\alpha I  & 0 \\
0 & -  I
\end{bmatrix}$. Therefore, the learning of the dissipativity parameters in~\eqref{eq_parametric_diss} is cast as the following optimization problem: 
\begin{align}\label{eq_obj_parametric}
\min_{P, Q, \rho, \alpha} & \ \ \Bigg\| \begin{bmatrix}
Q  & 0 \\
0 & P
\end{bmatrix} \Bigg\|_F^2 - \rho  + \lambda \alpha \\
 s. t. & \ \ P \succeq \epsilon I,\\
 & \ \ Q \prec  \begin{bmatrix}
\alpha  & 0 \\
0 & - I
\end{bmatrix}, \alpha > 0, \\
 & \ \ \biggl \langle \begin{bmatrix}
Q  & 0 \\
0 & P
\end{bmatrix}, \begin{bmatrix}
\Gamma_t  & 0 \\
0 & \Delta_t
\end{bmatrix}  \biggr \rangle \geq \rho, 1 \leq t \leq T,  
\end{align}
where $\Gamma_t$ and $\Delta_t$ are in~\eqref{eq_dualDissipativityParameter}, $\lambda$ serves as a weighting parameter that characterizes the trade-off between the dissipativity margin and the ${\cal L}_2$-gain. A larger $\lambda$ penalizes higher $\alpha$ more strongly, thereby enforcing a stricter ${\cal L}_2$-gain, that favors systems with smaller input–output amplification. Conversely, a smaller $\lambda$ relaxes this restriction, prioritizing the feasibility margin $\rho$ and fit the dissipativity inequalities closely to the data. Hence, $\lambda$ controls the trade-off between conservativeness and data adaptivity. Large $\lambda$ yields a more conservative estimate of dissipativity parameters, while smaller $\lambda$ leads to tighter data fitting at the risk of weaker generalization.

\section{Dissipativity learning in RKHS}\label{sec_RKHS}
The parametric dissipativity learning framework in section~\ref{sec_parameterlearning} constrains the supply and storage function $s$ and $V$, to have quadratic forms as in~\eqref{eq_def_QandP}.
Although this yields a tractable formulation, it limits expressiveness and may inadequately capture the nonlinear behavior of complex systems. To address this limitation, the framework is extended to {\it Reproducing Kernel Hilbert Spaces} (RKHSs) in this section, where $s$ and $V$ are sought in rich, infinite-dimensional function spaces. The RKHS structure defines functions through inner products with kernel-induced feature maps, enabling flexible and data-driven representations of dissipativity properties while preserving convexity and computational tractability of the learning problem.
 
Hence, we first discuss various preliminary results necessary for learning dissipativity properties in RKHS in section~\ref{sec_RKHS}. The nonparametric dissipativity learning in RKHS is introduced in section~\ref{sec_RKHS_diss}.

\subsection{Reproducing Kernel Hilbert Spaces}\label{sec_RKHS}
We first introduce reproducing kernel Hilbert spaces with its properties. 
\begin{definition}\cite{def_RKHS_50}
    A reproducing kernel Hilbert Space (RKHS) is a Hilbert Space ${\cal H}_k$ of functions $f: {\cal X} \rightarrow \mathbb{R}$ with a reproducing kernel $k: {\cal X}^2 \rightarrow \mathbb{R}$ where $k(x, \cdot) \in {\cal H}$ and $f(x) = \<k(x, \cdot), f\>$.
\end{definition}

The RKHS is discussed in more details in the following. Consider a kernel function $k(x, \cdot)$ with two variables $x, \cdot \in {\cal X} \times {\cal X}$. RKHS is a function space consisting functions $f$ such that
${\cal H}_k \eqdef \{f(\cdot) = \sum_{i=1}^n\alpha_i k(x_i, \cdot)\}$.
This construction equips $\mathcal{H}_k$ with an inner product 
$\langle f, g \rangle_{\mathcal{H}_k}$ such that the reproducing property 
$f(x) = \langle k(x,\cdot), f \rangle_{\mathcal{H}_k}$ holds for all 
$f \in \mathcal{H}_k$ and $x \in \mathcal{X}$. 
This property implies that the evaluation of any function at a point $x$ 
can be expressed as an inner product with the kernel section $k(x,\cdot)$, 
that is crucial for learning problems to 
be carried out entirely in terms of kernel evaluations. 

Specifically, instead of assuming quadratic forms as in the parametric setting, we lift 
the data points $z$ and $x$ into (possibly infinite-dimensional) feature 
spaces $\mathcal{H}_z$ and $\mathcal{H}_x$ using feature map in the following assumption.
\begin{assumption}\label{assum_featuremaps}
    Suppose that feature maps that takes data into RKHS are $\phi: {\cal Z} \rightarrow {\cal H}_z$ and $\psi: {\cal X} \rightarrow {\cal H}_x$, respectively, where ${\cal H}_z$ and ${\cal H}_x$ are the corresponding RKHSs.
\end{assumption}
\begin{remark}
    The feature map enables linear operation in RKHSs to represent nonlinear functions in the original space.
    From the defining property of RKHSs, inner products between features $\phi(z)$ or $\psi(x)$ define similarities through the kernel as follows
    \be\label{eq14}
        \langle \phi(z), \phi(z') \rangle_{{\cal H}_z} = \kappa(z, z').
    \ee
\end{remark}
Similarly to~\eqref{eq_parametric_diss} and~\eqref{eq_dualDissipativityParameter}, we redefine the dual dissipativity parameters in the following.
\begin{definition}
    An input-output trajectory $\{ y_{t}, u_{t}\}_{t=1}^T$ that determines $z_{t} = [y_{t} \ u_{t}]^{\sf T}$, the corresponding dual dissipativity parameters refer to
    \be\label{eq_dualDissipativityParameter_RKHS}
    \begin{cases}
    \Gamma_t: \phi(z_{t}) \otimes \phi(z_{t})\\
    \Delta_t: \psi(x_{t}) \otimes \psi(x_{t})  - \psi(x_{t+1}) \otimes \psi(x_{t+1}),
    \end{cases}
    \ee
    where $\Gamma_t$ and $\Delta_t$ are Hilbert-Schmidt operators in ${\cal H}_z$ and ${\cal H}_x$, respectively.
\end{definition}
\begin{remark}\label{remark_Gamma_Delta}
    $\Gamma_t$ is a rank one operator in the space of Hilbert-Schmidt operators on ${\cal H}_z$. $\Delta_t$ is at most rank two operator.
\end{remark}
\begin{definition}
    The supply rate and storage functions are defined as 
    \be\label{eq_DissipativityParameter_RKHS}
    \begin{cases}
    &s(z) \eqdef \<\phi(z), \Pi \phi(z) \>\\
    &V(x) \eqdef \<\psi(x), P \psi(x) \>,
    \end{cases}
    \ee
    where $\Pi: {\cal H}_z \rightarrow {\cal H}_z$ and $P: {\cal H}_x \rightarrow {\cal H}_x$ are Hilbert-Schmidt operators and new dissipativity parameters.
\end{definition}

With the above assumption, $s(z)$ and $V(x)$ can be expressed as inner products 
with operators acting on these features, and regularization can be 
performed via the Hilbert--Schmidt norms of these operators. 
This construction preserves convexity of the learning problem 
while allowing highly nonlinear, data-adaptive representations 
of dissipativity, which leads directly to the formulation in 
Section~\ref{sec_RKHS_diss}.

\subsection{OC-SVM for Nonparametric Dissipativity Learning}\label{sec_RKHS_diss}
The RKHS framework generalizes the parametric quadratic forms 
by lifting the variables into high-dimensional feature spaces. The following theorem provides the dissipativity learning problem in the RKHS framework. 
\begin{theorem}\label{thm_dissipativity_obj}
    Suppose Assumption~\ref{assum_featuremaps} holds. The dissipativity in~\eqref{eq_def_dissipativity} is captured by the following inequality with margin $\rho$ such that
\begin{align}
\<\Pi, \Gamma_t\> + \< P, \Delta_t\> \geq \rho,
\end{align}
for all sampled data points.
\end{theorem}
\begin{proof}
From~\eqref{eq_dualDissipativityParameter_RKHS},~\eqref{eq_DissipativityParameter_RKHS} and Assumption~\ref{assum_featuremaps}, it follows that for all $t$, $s(z_t) + V(x_t) - V(x_{t+1}) \geq 0$, that is
\begin{align}\label{eq16}
\nonumber
& \langle \phi(z_{t}), \Pi \phi(z_{t})\rangle - \big( \left<\psi(x_{t+1}), P \psi(x_{t+1})\right> - \<\psi(x_{t}), P \psi(x_{t}) \>\big) \\
=&\textnormal{tr}\left( \Pi \phi(z_{t}) \otimes \phi(z_{t})\right) +  \textnormal{tr}\left(  P   \psi(x_{t}) \otimes \psi(x_{t})\right) - \textnormal{tr}\left(  P  \psi(x_{t}) \otimes \psi(x_{t})\right),
\end{align}
which can be written as in~\eqref{eq16} from the definition of $\Gamma_t$ and $\Delta_t$ in~\eqref{eq_dualDissipativityParameter_RKHS}. This completes the proof.
\end{proof}

Theorem~\ref{thm_dissipativity_obj} provides an inequality to be satisfied for dissipativity. We posed certain constraints on the operators to ensure the robustness of the dissipativity learning. 

Firstly, the operator $P$ that corresponds to the storage function $V(x) = \langle \psi(x), P\psi(x)\rangle_{{\cal H}_x}$, must be a positive operator in ${\cal H}_x$, i.e., $P \in {\textnormal{HS}}_+({\cal H}_x)$. This is enforced by constraining $P - \epsilon P_0 \geq 0, \epsilon > 0$, where $P_0$ can be an identity operator. This ensures $P$ to be strictly positive. Secondly, we introduce two operators $\Upsilon_0$ and $\Upsilon_0$, that act on the feature map $\phi(z)$ and jointly upper bound the supply rate as $s(z) \leq \<\phi(z ), \Upsilon_0 \phi(z )\> + \alpha \<\phi(z ), \Upsilon_1 \phi(z )\>, \alpha > 0$. Note that in this bound, we can specify that $\Upsilon_1$ acts on the output feature $\phi(y)$. To this end, $\alpha$ characterizes the bound for ${\cal L}_2$-gain. 

The representer theorem makes the operators 
$\Pi, P$ lie in the finite span of the training data 
$\{\varphi(z), \psi(x)\}$, thus reducing the infinite-dimensional 
optimization problem to a finite-dimensional convex program. 
This makes nonparametric dissipativity learning both expressive and tractable, 
allowing us to capture highly nonlinear energy-exchange properties 
of the system while maintaining computational feasibility.

Follow the structure of OC-SVM for learning dissipativity parameters in section~\ref{sec_OCSVM_diss_parameter}, 
we introduce a margin $\rho $ of the inequality of dissipativity,
$\rho > 0 $ and it yields an optimization problem for $ 1 \leq t \leq T$:
\begin{align}\label{eq_def_nonparam_obj}
\min_{\Pi, P, \rho,\alpha } & \frac{1}{2} \left( \|  \Pi \|_{\textnormal{HS}}^2 + \|  P \|_{\textnormal{HS}}^2 \right) - \rho + \lambda \alpha\\
\nonumber
 s. t. & \ \ P - \epsilon P_0 \geq 0, \\
 \nonumber
 & - \Pi - \Upsilon_0 + \alpha \Upsilon_1  \geq 0, \alpha > 0 \\
\nonumber
& \<\Pi, \Gamma_t\> + \<P, \Delta_t\> \geq \rho,
\end{align}
where $\Upsilon_0$ and $\Upsilon_1$ act on the features of $y$ and $u$, respectively. 
\subsection{The Finite-dim Optimization}
The dissipativity learning problem in~\eqref{eq_def_nonparam_obj} is posed in terms of operators $P, \Pi$ acting on infinite-dimensional Hilbert spaces, which makes the problem intractable for direct computation. To bridge this gap, we leverage the representer property of RKHS operators and restrict $P$ and $\Pi$ to finite-rank forms spanned by the available data, thereby expressing all operator norms and dissipativity constraints through kernel Gram matrices. This reduction replaces inner products in feature space with kernel evaluations, transforming trace and bilinear operator terms into quadratic forms involving $\Phi, \Psi, \Phi_y$, and $\Phi_u$. As a result, the original operator-based OCSVM formulation is equivalently rewritten as a finite-dimensional constrained optimization problem over the coefficient matrices $p$ and $\pi$, together with the slack margin $\rho$ and gain parameter $\alpha$. This step ensures that the infinite-dimensional dissipativity learning problem is now amenable to practical computation, while retaining the expressive power of the RKHS-based formulation. That is, the operators play a central role in both the cost function and the dissipativity constraints, and their treatment in RKHS allows one to bridge the infinite-dimensional formulation with a finite-dimensional optimization problem.
 
Specifically, to obtain a finite-dimensional optimization formulation of~\eqref{eq_def_nonparam_obj}, we interpret both the objective and the constraints in the RKHS setting.
\subsubsection{Operators in the cost}
We begin with the Hilbert-Schmidt (H-S) norm of an operator $\Pi: {\cal H}_z \to {\cal H}_z$.
\begin{definition}
    For a Hilbert-Schmidt operator $\Pi$ on a Hilbert space ${\cal H}_z$, its norm is defined as
    $\|  \Pi \|_{\textnormal{HS}}^2 \eqdef \sum_k \| \Pi e_k \|^2$,
    where ${e_k}$ is an orthonormal basis of the Hilbert space. Equivalently, if $\Pi$ admits a decomposition $\Pi =\sum_r u_r \otimes v_r$ with $u_r, v_r \in {\cal H}_z$, then
    \be\label{1016_1}
    \|  \Pi \|_{\textnormal{HS}}^2 = \sum_{r, s} \langle u_r, u_s \rangle \langle v_r, v_s \rangle. 
    \ee
\end{definition}

In practice, we restrict $\Pi$ and $P$ to a finite-dimensional subspace of ${\cal H}_z$ and ${\cal H}_x$ spanned by the training samples. That is, $\Pi$ and $P$ are assumed to be finite-rank operators of rank at most $T$, expressed as
$\Pi = \sum_{i,j}^T  \pi_{i,j} \, \phi_i \otimes \phi_j, \ \ P = \sum_{i,j}^T  p_{i,j} \, \psi_i \otimes \psi_j$,
where $\pi_{i,j}$ and $p_{i,j}$ are coefficients that correspond to the basis $\phi_i \otimes \phi_j$ and $\psi_i \otimes \psi_j$, respectively. 
Hence, by construction, the coefficients $\pi_{i,j}$ and the feature vectors $\phi(z_i)$ define $\Pi$ completely on the subspace spanned by $\{\phi(z_1), \cdots, \phi(z_T)\}$. Therefore, the Hilbert-Schmidt norm of $\Pi$ is as follows: 
\begin{align} 
 \|  \Pi \|_{\textnormal{HS}}^2 
= & \textnormal{tr}\left( {\pi}^{\sf T} \Phi \pi \Phi^{\sf T} \right), \label{0930_2} 
\end{align}
where
$\Phi$ is the Gram matrix constructed by data; $\Phi_{i,j}$ is the element on $i$-th row and $j$-th column, i.e. $\Phi_{i,j} \eqdef \kappa (\phi(z_i), \phi(z_j)) = \<\phi(z_i), \phi(z_j)\>$ as shown in~\eqref{eq14};~\eqref{0930_2} holds from the symmetry of $\Phi$. Similarly, the coefficients $p_{i,j}$ and the feature vectors $\phi(x_i)$ define $P$ completely on the subspace spanned by $\{\phi(x_1), \cdots, \phi(x_T)\}$ and it follows that
\begin{align}\label{1007_1} 
 \|  P \|_{\textnormal{HS}}^2 =\textnormal{tr}\left( {p}^{\sf T} \Psi p \Psi^{\sf T} \right).
\end{align}
\subsubsection{Operators in the dissipativity inequality constraint} For the dissipativity inequality $\<\Pi, \Gamma_t\> + \< P, \Delta_t\> \geq \rho$ as in Theorem~\ref{thm_dissipativity_obj}, we restrict the operators $\Gamma_t$ and $\Delta_t$ to be spanned by the training samples. From~\eqref{eq_dualDissipativityParameter_RKHS}, it follows that
\begin{align} 
\<\Pi, \Gamma_t \>
= &  \sum_{i,j}^T \pi_{i,j} \Phi_{j,t}\, \Phi_{i,t},\label{1016_5} 
\end{align}
where
~\eqref{1016_5} holds from the trace of a rank one operator equals the inner product of its factors, i.e., $\textnormal{tr}(\phi(z_i) \otimes \phi(z_t)) = \<\phi(z_t), \phi(z_i)\>$. Similarly, 
\begin{align}
\<P,\Delta_t\>  
= \sum_{i,j}^T p_{i,j} \left( \Psi_{j,t} \Psi_{i,t} - \Psi_{j,t+1} \Psi_{i,t+1} \right). \label{1016_6}
\end{align}
Therefore, the dissipativity inequality is 
\be\label{1004_4}
\sum_{i,j}^T \pi_{i,j} \Phi_{j,t}\, \Phi_{i,t} + \sum_{i,j}^T p_{i,j} \left( \Psi_{j,t} \Psi_{i,t} - \Psi_{j,t+1} \Psi_{i,t+1} \right) \geq \rho.
\ee
\subsubsection{Operators in the constraints for bounds of the storage and $L_2$ gain} 
In learning dissipativity parameters in~\eqref{eq_obj_parametric}, it is assumed that the storage is quadratic with baseline $V(x) \geq \epsilon V_0(x)$ with $V_0(x) \eqdef \|x\|^2$. Similarly, let $V(x) \eqdef \<\psi(x), P\psi(x)\>$ and $P_0 \eqdef \sum_i \sum_j c_{i} c_j \psi(x_i) \otimes \psi(x_j) \in {\cal H}_x$ such that $P \geq \epsilon P_0$. From~\eqref{eq_DissipativityParameter_RKHS}, it follows that for all $t$,
the constraint that corresponds to the energy storage is
\be
\begin{aligned}\label{1004_3}
& \<\psi(x_t), P\psi(x_t) \> - \<\psi(x_t), P_0\psi(x_t) \> \\
=  & \textnormal{tr}\left(P \psi(x_t) \otimes \psi(x_t)\right) - \textnormal{tr}\left(P_0 \psi(x_t) \otimes \psi(x_t)\right) \\
= & \textnormal{tr}\left(\sum_{i,j} p_{i,j}\psi(x_i)\otimes \psi(x_j) \psi(x_t) \otimes \psi(x_t)\right) - \textnormal{tr}\left(\sum_{i,j} c_i c_j\psi(x_i)\otimes \psi(x_j) \psi(x_t) \otimes \psi(x_t)\right) \\
= & \textnormal{tr}\left(\sum_{i,j} p_{i,j} \Psi_{j,t}\psi(x_i)\otimes \psi(x_t)\right)  - \textnormal{tr}\left(\sum_{i,j} c_i c_j \Psi_{j,t}\psi(x_i)\otimes \psi(x_t)\right)  \\
= & \sum_{i,j} (p_{i,j} - c_i c_j) \Psi_{j,t} \Psi_{i,t}.
\end{aligned}
\ee
To obtain the linear coefficients $c_i, i \in \{1,2,\cdots, T\}$ in the operator $P_0$, we assume that
\be\label{1006_1}
\sqrt{V_0(x)} =  \< \sum_{i=1}^T c_i\psi(x_i), \psi(x)\>,
\ee
that is, the function $V_0$ is approximated by a linear summation of
$(\sqrt{V_0})(\cdot) = \sum_{i=1}^T c_i \Psi(x_i,\cdot)$.
Therefore, with~\eqref{1006_1}, we admit the following assumption.
\begin{assumption}
For all $t = \{1,2,...,T\}$, the following holds for all $c_i$
\be\label{20250410_1}
\sqrt{V_0 (x_t)} = \sum_{i=1}^T c_i \Psi_{i,t} = \|x_t\|.
\ee
\end{assumption}
The coefficients $ c \eqdef \left[c_1, c_2, \cdots, c_T \right]^{\sf T} \in \mathbb{R}^T $ can be estimated as 
\be\label{eq_hat_c}
\hat{c} = \Psi^{-1} \left[\|x_1\|, \|x_2\|, \cdots, \|x_T\| \right]^{\sf T}.
\ee
 
For the supply function $s(z)$, let $\Upsilon_0, \Upsilon_1 \in {\cal H}_z$ be operators applying on the outputs $y$ and inputs $u$, respectively. 
Similar to~\eqref{20250410_1}, we assume that
\begin{align}
\|y_t\| \approx &\sum_{i=1}^T \upsilon_i^y \Phi^y_{i,t}, \ \ \|u_t\| \approx \sum_{i}^T \upsilon_i^u \Phi_{i,t}^u.
\end{align}
Consequently, we make the following kernel interpolation approximation scheme. For all $t = \{1,2,...,T\}$, 
    \begin{align}\label{1004_1}
\!\!\|y_t\|^2 \!\!= \!\!\sum_{i, j}^T \!\upsilon_i^y \upsilon_j^y \Phi^y_{i,t} \Phi^y_{j,t}, \|u_t\|^2 \!\!= \!\!\sum_{i, j}^T\! \upsilon_i^u \upsilon_j^u \Phi^u_{i,t} \Phi^u_{j,t}. \ \ 
\end{align}

The coefficient estimates are 
\begin{align}\label{1006_2}
\hat{v}^y = &(\Phi^y)^{-1} \left[\|y_1\|, \|y_2\|, \cdots, \|y_T\| \right]^{\sf T}, \\
\hat{v}^u = &(\Phi^u)^{-1} \left[\|u_1\|, \|u_2\|, \cdots, \|u_T\| \right]^{\sf T},
\end{align}
where $\Phi^y_{ij} \eqdef \kappa_y(y_i, y_j)$ and $\Phi^u_{ij} \eqdef \kappa_u(y_i, y_j)$.
Hence, the estimate $\hat{\Upsilon}_0$ and $\hat{\Upsilon}_1$ are both finite-rank Hilbert-Schmidt operators in the corresponding RKHS such that
\begin{align}\label{1007_2}
\hat{\Upsilon}_0 
&= \sum_{i=1}^{T} \sum_{j=1}^{T} 
\hat{v}_i^{\,y} \, \hat{v}_j^{\,y} \;
\phi(y_i) \otimes \phi(y_j), \\ 
\hat{\Upsilon}_1 
&= \sum_{i=1}^{T} \sum_{j=1}^{T} 
\hat{v}_i^{\,u} \, \hat{v}_j^{\,u} \;
\phi(u_i) \otimes \phi(u_j).
\end{align}
Note that the estimated $\hat{\Upsilon}_0$ and $\hat{\Upsilon}_1$ are constrained to the available data, which causes estimation error with respect to $ {\Upsilon}_0$ and $ {\Upsilon}_1$, respectively.

The $L_2$ gain-like inequality enforces that $- \hat{\Pi} + \hat{\Upsilon}_0 + \alpha \hat{\Upsilon}_1 \geq 0$, and for all $t$,
\begin{align}\label{eq36}
-  & \< \phi(y_t), \hat{\Upsilon}_0\phi(y_t)\> +  \alpha \< \phi(u_t), \hat{\Upsilon}_1\phi(u_t) \>  \geq \<  \phi(z_t), \hat{\Pi} \phi(z_t) \>.
\end{align}
Hence, the following holds
\begin{align}\label{1004_2}
\begin{aligned}
\sum_{i,j}^T \!\!\pi_{i,j}\! \Phi_{j,t} \!\Phi_{i,t} \!\!
\leq \!- \!\!\sum_{i,j}^T \hat{\upsilon}_i^y \hat{\upsilon}_j^y \Phi_{i,t}^y \Phi^y_{j,t} \!\!+ \! \!\alpha \!\!\sum_{i,j}^T\! \hat{\upsilon}_i^u\! \hat{\upsilon}_j^u \!\Phi_{i,t}^u \Phi^u_{j,t}. \ \ \ \ \   
\end{aligned}
\end{align}

Based on the finite-rank operator expansions
the Hilbert--Schmidt operators $P$ and $\Pi$ are restricted to the finite spans of the training data, 
so that their actions and norms can be expressed entirely in terms of the kernel Gram matrices $\Phi$, $\Phi^y$, $\Phi^u$ and $\Psi$. 
Specifically, the Hilbert--Schmidt norms of these operators become quadratic forms of the coefficient matrices $p$ and $\pi$, 
yielding the trace expressions in~\eqref{0930_2} and~\eqref{1007_1}, respectively. 

The dissipativity inequalities involving $\Pi$, $P$, $\Gamma_t$, and $\Delta_t$ 
reduce to algebraic constraints represented through kernel evaluations as in~\eqref{1016_5} and \eqref{1016_6}. 
The operator bounds enforcing the positivity of the storage function 
and the $\mathcal{L}_2$-gain condition are similarly reformulated through data-dependent terms 
as shown in~\eqref{1004_3} and~\eqref{1007_2}, 
where the reference operators $\hat{\Upsilon}_0$ and $\hat{\Upsilon}_1$ 
are constructed from the regression coefficients $\hat{v}^{y}$ and $\hat{v}^{u}$, 
and the baseline storage operator is determined by $\hat{c}$. 

From these results, the infinite-dimensional operator optimization problem in the RKHS 
is reduced to the following finite-dimensional convex optimization problem:
\be
\begin{aligned} 
\min_{p, \pi, \rho, \alpha } & \ \ \frac{1}{2} \left( \textnormal{tr}\left( {\pi}^{\sf T} \Phi \pi \Phi^{\sf T} \right) +  \textnormal{tr}\left( {p}^{\sf T} \Psi p \Psi^{\sf T} \right) \right) - \rho  + \lambda \alpha\\
\textnormal{s.t.} & \ \ \rho \geq 0, \alpha \geq 0, \\
& \sum_{i,j}^T (p_{i,j} - \hat{c}_i \hat{c}_j) \Psi_{j,t} \Psi_{i,t} \geq 0,\\
&  \sum_{i,j}^T\!\! \alpha \hat{\upsilon}_i^u \hat{\upsilon}_j^u  \Phi_{i,t}^u \Phi^u_{j,t} \!-\!\hat{\upsilon}_i^y \hat{\upsilon}_j^y \Phi_{i,t}^y \Phi^y_{j,t} \geq \!\!\sum_{i,j}^T \!\!\pi_{i,j}  \Phi_{j,t}  \Phi_{i,t},\\
& \sum_{i,j}^T \pi_{i,j} \Phi_{j,t} \Phi_{i,t} \! + \! p_{i,j} \!
\left( \Psi_{j,t} \Psi_{i,t} \!\!- \!\!\Psi_{j,t+1} \Psi_{i,t+1} \right) \geq \rho.\\
\end{aligned} 
\ee
Note that the parameters $\hat{c}$, $\hat{v}^{u}$, and $\hat{v}^{y}$ are obtained from the training data samples according to~\eqref{eq_hat_c} and~\eqref{1006_2}.

\section{Error Analysis for Learning in RKHS}\label{sec_ErrorAnalysis}
In KRR, we assume that data samples $\bar{z}_i, i \in \{1,2,\cdots, m\}$ are drawn from a probability measure $\rho$ on ${\bar{\cal{Z}}} = {\cal \bar{X}} \times {\cal \bar{Y}}$. The regression function to be solved is
\be\label{eq46}
f_{\rho}(x) = \int_{Y}yd\rho(y|x), x \in {\cal \bar{X}},
\ee
where $\rho(y|x)$ is the conditional distribution at $x$ induced by $\rho$. 

Let $K: {\cal \bar{X}} \times {\cal \bar{X}} \to \mathbb{R}$ be a {\it Mercer kernel}, that is, the matrix formed by $K_{i,j} = K(\bar{x}_i, \bar{x}_j), i, j \in \{1, 2, \cdots, m\}$ is positive semidefinite for any finite set of distinct data points $\bar{x}_i, i \in \{1,2,\cdots, m\}$. We denote the generated RKHS by kernel $K$ as ${\cal H}_K$. Hence, the regression algorithm we investigate is a Tikhonov regularization scheme associated with $K$, that is, 
\be\label{eq_def_TikhonovLearning}
f_{z,\lambda}\!\eqdef\!\argmin_{f\in {\cal H}_K} \left\{ \frac{1}{m} \sum_{i=1}^m (f(x_i) \!-\! y_i)^2 + \lambda \|f\|_K^2 \!\!\right\}
\ee

The following theorem provides an upper bound for approximating $f_\rho$ through the learning scheme in~\eqref{eq_def_TikhonovLearning} under the condition that $f_\rho$ is in the range of $L_K^{r}$.
\begin{theorem}~\cite{SZ_CA_07}\label{theorem_error_bound} 
	Let $\bar{z}$ be randomly drawn according to $\rho$ and satisfy $|\bar{y}| \le M$ almost surely. 
	Assume that, for some $0 < r \le 1$, $f_\rho$ is in the range of $L_k^{r}$. 
	Take the regularization parameter as
	$\lambda 
	= \left( 3 \kappa M/{\| L_K^{-r} f_\rho \|_{\rho}} \right)^{2/(1 + 2r)}
	\, m^{-1/(1 + 2r)}$, where $\kappa \eqdef \sqrt{\textnormal{sup}_{x\in {\bar{\cal X}}}K(x, x)}$.
	Then, for any $0 < \delta < 1$, with confidence $1 - \delta$, 
	\begin{equation}
		\label{eq44}
		\begin{aligned}
			&\| f_{z, \lambda} - f_\rho \|_{K}\\
			&\le  4 \log\!\left(\frac{2}{\delta}\right)
			(3 \kappa M)^{\frac{2r - 1}{2r + 1}} \| L^{-r}_{K} f_\rho \|_{\rho}^{2/(1+2r)}m^{-\frac{2r-1}{4r+2}}.
		\end{aligned}
	\end{equation}
\end{theorem}

The above Theorem provides a quantitative learning rate for the KRR approximation in the RKHS. In our context, the operators $\Upsilon_0, \Upsilon_1$ are estimated through such regressions as in~\eqref{1007_2}. 

The function $f_{\rho}$ that we are approximating is $f(\cdot) = \| \cdot \|^2$ as in~\eqref{eq39}. 
In the proposed dissipativity learning in RKHS, we estimate the Hilbert-Schmidt operators $\Upsilon_0$ and $\Upsilon_1$ by Kernel ridge regression (KRR) so that $\< \phi(y), \Upsilon_0\phi(y)\> = \|y\|^2$ and $\< \phi(u), \Upsilon_1\phi(u)\> = \|u\|^2$. Let the learned operators be $\hat{\Upsilon}_0$ and $\hat{\Upsilon}_1$ and suppose that, for some norms on functions of $y$ and $u$
\begin{align}\label{eq39}
    \left\lVert\< \phi(y), \hat{\Upsilon}_0\phi(y)\> - \|y\|^2\right\rVert \leq & \epsilon_0, \\
    \left\lVert\< \phi(u), \hat{\Upsilon}_0\phi(u)\> - \|u\|^2\right\rVert \leq & \epsilon_1.
\end{align}

Note that $\Gamma_t \eqdef \phi(z_{t}) \otimes \phi(z_{t}) \in {\textnormal{HS} }({\cal H}_z) $ and $\Delta_t \eqdef \psi(x_{t}) \otimes \psi(x_{t})  - \psi(x_{t+1}) \otimes \psi(x_{t+1})  \in {\textnormal{HS} }({\cal H}_x) $.
Denote by $(\hat{\Pi}, \hat{P}, \hat{\rho})$ any feasible solution of the optimization so that 
\be\label{eq444}
\<\hat{\Pi}, \Gamma_t\> + \<\hat{P}, \Delta_t\> \geq \hat{\rho}, \forall t.
\ee
Then, when generalized to the entire population, $\Gamma$ and $\Delta$ are random operators in ${\cal H}_z$ and ${\cal H}_x$, respectively. We assume that with probability $1-\delta$, 
\begin{align}\label{eq41} 
\<\hat{\Pi}, \Gamma\> + \<\hat{P}, \Delta\> \geq \hat{\rho} - \epsilon_{\rho}, 
\end{align}
for some generalization error $\epsilon_{\rho} \geq 0$. 
Therefore, from~\eqref{eq36} and~\eqref{eq39}, the propagation of operator errors to the supply bound is
\be\label{eq40}
\begin{aligned}
& \<  \phi(z_t), \hat{\Pi} \phi(z_t) \> \\
\leq  &- \< \phi(y_t), \hat{\Upsilon}_0\phi(y_t)\> +  \alpha \< \phi(u_t), \hat{\Upsilon}_1\phi(u_t) \>  
\leq   -\|y\|^2 +\alpha \|u\|^2 + \epsilon_0 + \epsilon_1.
\end{aligned}
\ee

Assume $\|\Gamma\|\geq 1$ and~\eqref{eq41} holds with probability $1-\delta$. 
Let
\be\label{eq43}
\Pi = \hat{\Pi} + \max(\hat{\rho} - \epsilon_{\rho}, 0)I,
\ee
where $I$ is the identity on ${\cal H}_z$. Then,
\be\label{eq44}
\< \Pi, \Gamma\> + \<\hat{P}, \Delta \> \geq 0.
\ee
Combining~\eqref{eq44} with~\eqref{eq40} yields that, for any $(y, u)$,
\be\label{eq_def_totalerror}
\begin{aligned}
\<  \phi(z),  \Pi  \phi(z) \> \leq  
-\|y\|^2 +\alpha \|u\|^2 + \bar{\epsilon},
\end{aligned}
\ee
where $\bar{\epsilon} \eqdef \epsilon + \max(\hat{\rho} - \epsilon_{\rho}, 0)$.
\begin{remark} 
The total error $\bar{\epsilon}$ in~\eqref{eq_def_totalerror} decomposes into: (i) approximation errors
$\epsilon$ from approximating via KRR, and (ii) the generalization error $\epsilon_\rho$ 
from empirical margin $\hat{\rho}$~in \eqref{eq444} to a population guarantee.
\end{remark}

The following theorem provides the total error bound for $\bar{\epsilon}$.
\begin{theorem}
    Let Assumption~\ref{assump_general_K_theta},~\ref{assum_f_theta_in_range},~\ref{assum_f_theta_lim} holds. Let $\bar{z}$ be randomly drawn according to $\rho$ and satisfy $|\bar{y}| \le M$ almost surely. For any $1<\delta<1$, with confidence $1-\delta$, the total approximation error of the proposed dissipativity learning in RKHS is bounded as
    \begin{align}
    \nonumber
    \bar{\epsilon} \leq & \inf_{\epsilon > 0}\{2\epsilon + 4 \log\!\left(\!\frac{2}{\delta}\!\right)\!
(3 \kappa M)^{\frac{1}{3}} c_{K} c_{f, s} \|l_{\pi, s, \theta }\|_{L^2}  m^{-\frac{1}{6}}\} + \frac{1}{m} \,
\mathcal{O}\!\left(
    \log \frac{1}{\delta}
      +   \epsilon^{\frac{1}{2}} \log m
      +   \frac{1}{\epsilon^2} \log \frac{1}{\epsilon}
\right)\!.
    \end{align}
\end{theorem}
\begin{proof}
    The proof follows from characterizing the error bound for KRR and the generalization error in Section~\ref{sec_error1} and in Section~\ref{sec_error2}, respectively. 
\end{proof}
 
\subsection{Approximation Error from KRR}\label{sec_error1}
\begin{assumption}\label{assump_general_K_theta}
	Assume that $K_{\theta}(x, y)$ is a radial kernel and $\int_{\mathbb{R}^d}K(x)dx = 1$
\end{assumption}
\begin{assumption}\label{assum_f_theta_in_range}
	Suppose that (1) $\left| {\cal F}\{f\}(\xi )\right| \leq \frac{c_{f, s}}{(1+ |\xi|^2)^{s/2}}$ almost everywhere $\forall \xi \in \mathbb{R}^n$; (2) there exists $l_{\pi, s, \theta} \in L^2(\mathbb{R}^d)$, such that $\int_{\mathbb{R}^d} \frac{\left|{\cal F}\{\pi^{\theta}\}(\xi -\eta)\right|}{(1+\|\eta\|^2)^{s/2}}d\eta \leq \left|{\cal F}\{K\}(\xi)\right| \cdot \left|l_{\pi, s, \theta}(\xi)\right|$ holds; (3) $\sup_{\xi \in \mathbb{R}^d}\left|{\cal F}\{ K_\theta\} (\xi) \right| = c_{K} < \infty$.
\end{assumption}

\begin{assumption}\label{assum_f_theta_lim}
	Suppose that (1)$\int_{|y-x| > \eta_\theta} |K_\theta(x-y)|dy < \frac{\epsilon}{6M_f M_\pi^\theta}$, where $M_f \eqdef \sup_{x \in \textnormal{supp}\pi} |f(x)|$, $M_\pi^\theta = \textnormal{sup}_{x\in \mathbb{R}^d} |\pi(x)|$; (2) when $\theta$ is sufficiently small and $|y-x| \leq \eta_\theta$, $|f(y) - f(x)|\leq L_f \eta_\theta$ holds for all $x \in \textnormal{supp} \pi$ and $\left|(\pi(y))^\theta - (\pi(x))^\theta\right| \leq L_\pi \eta_\theta$ holds for all $x \in \textnormal{supp} f$; (3) $\int_{\mathbb{R}^d} |K(t)| dt \leq M_K$; (4) $\forall x \in \textnormal{supp} f$, $|\pi(x)| \in [m_\pi, M_\pi]$, and $\theta \leq \{\frac{\log (1- \frac{\epsilon}{3M})}{\log m_\pi}, \frac{\log (1+ \frac{\epsilon}{3M})}{\log M_\pi}\}$, $M \geq \left|f(x)\right|$.
\end{assumption}

The following theorem provides the approximation error from KRR.
\begin{theorem}\label{theorem_error_bound_KRR}
    Let Assumption~\ref{assump_general_K_theta},~\ref{assum_f_theta_in_range},~\ref{assum_f_theta_lim} holds. Let $\bar{z}$ be randomly drawn according to $\rho$ and satisfy $|\bar{y}| \le M$ almost surely. Then, $\forall \epsilon > 0$, $\exists \theta_\epsilon \in (0, 1)$ such that $f_{\theta_\epsilon} \in \operatorname{Range}(L_K)$ and
    $ \lim_{\theta_\epsilon \to 0} \|f_{\theta_\epsilon} - f \|_{\infty} = \epsilon$ hold. 
Moreover, for any $0 < \delta < 1$, with confidence $1 - \delta$, 
  the KRR approximation error is bounded as follows:
    \begin{align}\label{eq_KRR_bound}
    &\|f - \hat{f}^{\textnormal{KRR}}\|_\infty \\
    \nonumber
    &\leq \inf_{\epsilon > 0}\{2\epsilon + 4 \log\!\left(\!\frac{2}{\delta}\!\right)\!
(3 \kappa M)^{\frac{1}{3}} c_{K} c_{f, s} \|l_{\pi, s, \theta }\|_{L^2}  m^{-\frac{1}{6}}\},
\end{align}
where $\kappa \eqdef \sqrt{\textnormal{sup}_{x\in {\bar{\cal X}}}K(x, x)}$. 
\end{theorem}
\begin{proof}
    The function that we are approximating is $f(\cdot) = \| \cdot \|^2$ as in~\eqref{eq39}. Suppose $f(x) \in \operatorname{Range}(L_K^{r})$. Then, there exists $\varphi \in L_{\rho}^2$ such that $f(x) = \int_{\mathbb{R}^d}K(x, y)\varphi(y)d\rho(y)$. Equivalently, $f \in {\cal C}^s$ follows from $K(\cdot, \cdot) \in {\cal C}^s$. Note that $f(x) = \|x\|^2 \in {\cal C}^1$ holds for all $K(\cdot, \cdot)$. Hence, $f(x) \notin \operatorname{Range}(L_K^{r})$.
    Consequently, the induced approximation error $\epsilon_0, \epsilon_1$ when approximating $\Upsilon_0$ and $\Upsilon_1$ in~\eqref{eq39} are not fully aligned with Theorem~\ref{theorem_error_bound}. 
    To this end, we break the proof into two parts. Firstly, we construct a parameterized function $f_\theta$ such that $f_\theta \in \operatorname{Range}(L_K^{r})$, and then, ensure that the function $f_\theta$ satisfy $\lim_{\theta \to 0} \|f_{\theta} - f \|_{\infty} = \epsilon, \forall \epsilon > 0$. 
    From the definition of an integral operator $L_K: {\cal H}_K \rightarrow {\cal H}_K$
\be\label{eq_def_operatorL_K}
L_K(f_{\rho})(x) \eqdef \int_{\cal X} K(x, y)f_{\rho}(y)d\rho_X(y),
\ee
the range of $L_k^{r}$ is
\begin{equation}\label{eq_def_rangeLk}
\operatorname{Range}(L_K^{r}) \eqdef
\{\, f = L_K^{r} g \;|\; g \in \mathcal{H}_K \,\}.
\end{equation}
Thus, we aim to construct 
\be\label{eq_def_f_theta}
\begin{aligned}
    & f_\theta \eqdef L_{K_\theta} f \\
    & \textnormal{s.t.} \ \ f_\theta\in \operatorname{Range}(L_K),\\
    & \lim_{\theta \to 0} \|f_{\theta} - f \|_{\infty} = \epsilon, \forall \epsilon > 0
\end{aligned}
\ee
where 
$L_{K_\theta} f \eqdef \int  K_{\theta}(x, y)\, f(y) d \rho(y)$, 
$d\rho(x) = (\pi(x))^\theta dx, \theta \in (0, 1]$, and $\pi(x)\geq 0$ is measurable on $\mathbb{R}^d$.
Then, the approximation error is bounded as follows:
    \be\label{eq_total_error}
    \|f -\hat{f}^{\textnormal{KRR}}\|_{\infty} \leq \|f-f_\theta\|_{\infty} + \|\hat{f}-\hat{f}_\theta\|_{\infty}+ \|f_\theta -\hat{f}_\theta\|_{\infty}.
    \ee
    Particularly, the assumptions such that $f_\theta\in \operatorname{Range}(L_K)$ and $\lim_{\theta \to 0} \|f_{\theta} - f \|_{\infty} = 0$ hold in Section~\ref{sec_f_theta_range} and Section~\ref{sec_f_theta_limit}, respectively. Then, we jointly evaluate a resulting bound for~\eqref{eq_total_error} in~\eqref{eq_KRR_bound}. This completes the proof.
\end{proof}

The following Corollary is for the case when the Kernel for KRR approximation is a Gaussian kernel. 
\begin{corollary}
    A Gaussian kernel satisfies 
        $\|l_{\pi, s, \theta_\epsilon}\| = C_1 \theta_\epsilon^{-d} \Pi_{i=1}^d \int_{\mathbb{R}} \left|\textnormal{exp}\left( -\frac{-\sigma^2}{\theta_\epsilon}(\xi_i - i)^2 + h^2 \xi_i^2\right)\right| d\xi_i$, that is $\|l_{\pi, s, \theta_\epsilon}\| \leq C_2 \frac{1}{\theta_\epsilon^d} \left(\frac{\sigma^2}{\theta_\epsilon} h^2 \right)^{-d/2} \leq C_3 \theta_\epsilon^{-d/2}  \leq C_4 \epsilon^{-d/2}.$ The corresponding KRR approximation error is
        \be
        \|f -\hat{f}^{\textnormal{KRR}}\|_\infty \leq \inf_{\epsilon > 0} \{2\epsilon + \frac{C \log(2/\delta) m^{-1/6}}{\epsilon^{d/2}}\}.
        \ee
\end{corollary}
\subsubsection{$f_\theta\in \operatorname{Range}(L_K)$}\label{sec_f_theta_range}
From the definition of integral operator as in~\eqref{eq_def_operatorL_K}, similarly, we have 
\be
L_{K_\theta} f \eqdef \int_{\mathbb{R}^d}  K_{\theta}(x, y)\, f(y)(\pi(y))^\theta dy,
\ee
for a parameterized kernel $K_\theta$.
Given $f \in L_{\rho}^2$, i.e., $\int |f(x)|^2d\rho(x) < \infty$, then $\int |f(x)\pi^{\theta/2}(x)|^2dx < \infty$, that is, $f\pi^{\theta/2} \in L^2$ and $f = \bar{f}\pi^{-\theta/2}, \bar{f} \in L^2$.
We evaluate the Fourier transform of $L_{K_\theta}f$, that is,
\be
\nonumber
\begin{aligned}
\mathcal{F}\!\left\{ L_{K_\theta} f \right\}
= & \mathcal{F}\!\left\{ \int_{\mathbb{R}^d} K_{\theta}(x, y)\, f(y)(\pi(y))^\theta d\rho(y) \right\}\\
= & \mathcal{F}\!\left\{ K_{\theta} \star (f \pi^\theta)\right\}\\
= & \mathcal{F}\{K_{\theta}\}\cdot \mathcal{F}\{f\pi^\theta\} \\
= & \mathcal{F}\{K_{\theta}\}\cdot \mathcal{F}\{\bar{f}\pi^{\theta/2}\}.
\end{aligned}
\ee
Hence 
$L_{K_\theta}f = {\cal F}^{-1} \left\{ \mathcal{F}\{K_{\theta}\}\cdot \mathcal{F}\{\bar{f}\pi^{\theta/2}\}\right\}$. Similarly, we obtain
\begin{align}
L_K^{-1} L_{K_\theta} f  
= & {\cal F}^{-1}\left\{ \frac{{\cal F} \left\{ L_{K_\theta} f \right\}}{{\cal F}\{K\}} \right\} \cdot \pi^{-\theta/2},\label{1016_7} \\
=& {\cal F}^{-1}\left\{ \frac{\mathcal{F}\{K_{\theta}\}\cdot \mathcal{F}\{\bar{f}\pi^{\theta/2}\}}{{\cal F}\{K\}} \right\} \cdot \pi^{-\theta/2},\label{1016_8} 
\end{align}
where~\eqref{1016_8} holds from plugging $L_{K_\theta}f = {\cal F}^{-1} \left\{ \mathcal{F}\{K_{\theta}\}\cdot \mathcal{F}\{\bar{f}\pi^{\theta/2}\}\right\}$ into~\eqref{1016_7}. The sufficient and necessary condition for $L_K^{-1} L_{K_\theta} f \in L_{\rho}^2$ is $L_K^{-1} L_{K_\theta} f \cdot \pi^{\theta/2}  \in L^2$, that is, 
\be\label{1016_9}
{\cal F}^{-1}\left\{ \frac{\mathcal{F}\{K_{\theta}\}\cdot \mathcal{F}\{\bar{f}\pi^{\theta/2}\}}{{\cal F}\{K\}} \right\} \in L^2.
\ee
\begin{prop}\label{prop_f_theta_in_range}
    If Assumption~\ref{assum_f_theta_in_range} holds, then, $f_\theta \in \operatorname{Range}(L_K)$ and it holds that
    $\|L_K^{-1}f_\theta\|_{L_\rho^2} \leq c_{K} c_{f, s} \|l_{\pi, s, \theta }\|_{L^2}, l_{\pi, s, \theta} \in L^2(\mathbb{R}^d)$.
\end{prop}
\begin{proof}
Given $\bar{f} = f \pi^{\theta/2}$, the functions $\bar{f}\pi^{\theta/2} = f\pi^\theta $ almost everywhere, hence, 
\begin{align}
\nonumber
\left| {\cal F} \{ \bar{f} \pi^{\theta/2}\}(\xi )\right| = & \left|{\cal F} \{ f \pi^\theta \}(\xi )\right| \\
\nonumber
=&\left|{\cal F} \{ f\} \star {\cal F}\{ \pi^\theta \}(\xi )\right| \\
\nonumber
=& \left| \int_{\mathbb{R}^d} {\cal F} \{ f\} (\eta) {\cal F} \{ \pi^\theta \}(\xi -\eta) d\eta  \right| \\
\nonumber
\leq & \int_{\mathbb{R}^d} \left|{\cal F} \{ f\} (\eta) \right| \cdot \left|{\cal F} \{ \pi^\theta \}(\xi -\eta)\right| d\eta   \\
\leq & \int_{\mathbb{R}^d} \frac{c_{f, s}}{(1+ |\xi|^2)^{s/2}} \left|{\cal F} \{ \pi^\theta \}(\xi -\eta)\right| d\eta \label{1017_1} 
\end{align}
where~\eqref{1017_1} holds from Assumption~\ref{assum_f_theta_in_range} (1). By moving $c_{f, s}$ to the left of~\eqref{1017_1}, 
\begin{align}
\frac{1}{c_{f,s}} \left|{\cal F}\{\bar{f}\pi^{\theta/2} \}(\xi)\right| \leq& \int_{\mathbb{R}^d} \frac{1}{(1+ |\xi|^2)^{s/2}} \left|{\cal F} \{ \pi^\theta \}(\xi -\eta)\right| d\eta 
 \leq  \left|{\cal F}\{K\}(\xi)\right| \cdot \left| l_{\pi, s, \theta}(\xi)\right|,\label{1020_3}
\end{align}
for a function $l_{\pi, s, \theta} \in L^2$, where~\eqref{1020_3} holds from Assumption~\ref{assum_f_theta_in_range} (2). Then,
$\left|\frac{{\cal F} \{\bar{f}\pi^{\theta/2}\}}{{\cal F}\{K\}}(\xi) \right| \leq \left|l_{\pi, s, \theta}(\xi)\right| \cdot c_{f, s}$.
From Assumption~\ref{assum_f_theta_in_range} (3), it follows
\begin{align}
\nonumber
 & \left \lVert \frac{\mathcal{F}\{K_{\theta}\}\cdot \mathcal{F}\{\bar{f}\pi^{\theta/2}\}}{{\cal F}\{K\}}\right \rVert_{L^2} 
 \leq \sup_{\xi \in \mathbb{R}^d} \left| {\cal F}\{ K_{\theta}(\xi)\}\right| \cdot \left \lVert \frac{ \mathcal{F}\{\bar{f}\pi^{\theta/2}\}}{{\cal F}\{K\}}\right \rVert_{L^2} 
 \leq  c_{K}c_{f,s} \| l_{\pi, s, \theta}(\xi) \|_{L^2}  
\end{align}
From Plancherel theorem, the above is equivalent to ${\cal F}^{-1}\left\{ \frac{\mathcal{F}\{K_{\theta}\}\cdot \mathcal{F}\{\bar{f}\pi^{\theta/2}\}}{{\cal F}\{K\}} \right\} < \infty$, that is,~\eqref{1016_9} holds, which implies $L_K^{-1} L_{K_\theta} f \in L^2$ and $f_\theta \in \operatorname{Range}(L_K)$. This completes the proof. 
\end{proof}
 
\subsubsection{$\lim_{\theta \to 0} \|f_{\theta} - f \|_{\infty} = 0$}\label{sec_f_theta_limit}
\begin{prop}\label{prop_f_theta_lim}
    If Assumption~\ref{assum_f_theta_lim} holds, then $\lim_{\theta \to 0} \|f_{\theta} - f \|_{\infty} = \epsilon$
\end{prop}
\begin{proof}
    From the definition of $f_\theta$, 
    \begin{align}
    \nonumber
        &\left|(L_{K_\theta}f -f)(x)\right| \\
        \nonumber
        = & \left|  \int_{\mathbb{R}^d}  K_{\theta}(x-y)\, f(y)(\pi(y))^\theta dy - f(x)\right| \\
        \nonumber
        \leq & \left|  \int_{\mathbb{R}^d}  K_{\theta}(x-y)\, f(y)(\pi(y))^\theta dy - f(x)(\pi(x))^\theta \right|  + \left|f(x)\right| \left|1-(\pi(x))^\theta \right| \\
        \leq & \int_{\mathbb{R}^d}  \left|K_{\theta}(x-y)\right| \left| f(y)(\pi(y))^\theta - f(x)(\pi(x))^\theta \right|dy   + \left|f(x)\right| \left|1-(\pi(x))^\theta \right| \label{eq67} \\
        \nonumber
        \leq & \int_{|y-x| > \eta_\theta}  \left|K_{\theta}(x-y)\right| \left| f(y)(\pi(y))^\theta - f(x)(\pi(x))^\theta \right|dy \\
        & + \int_{|y-x| \leq \eta_\theta}\!\!\!\!\!\!\!\left|K_{\theta}(x-y)\right| \left| f(y)(\pi(y))^\theta - f(x)(\pi(x))^\theta \right|dy   + \left|f(x)\right| \left|1-(\pi(x))^\theta \right|,\label{1020_1} 
    \end{align}
    where~\eqref{eq67} holds from $\int_{\mathbb{R}^d}K(x)dx = 1$ in Assumption~\ref{assump_general_K_theta}.
We start with the first term in~\eqref{1020_1}. From Assumption~\ref{assum_f_theta_lim} (1), it follows 
\begin{align}
\nonumber
    & \int_{|y-x| > \eta_\theta}  \left|K_{\theta}(x-y)\right| \left| f(y)(\pi(y))^\theta - f(x)(\pi(x))^\theta \right|dy \\
    \nonumber
     < & \frac{\epsilon}{6M_f M_\pi^\theta}\cdot 2M_fM_\pi^\theta \\
    = & \epsilon/3.\label{1020_condi1}
\end{align}
 
For the second term in~\eqref{1020_1}, note that from Assumption~\ref{assum_f_theta_lim} (2), 
\begin{align}
\nonumber
    & \left| f(y)(\pi(y))^\theta - f(x)(\pi(x))^\theta \right| \\
    \nonumber
    \leq  & \left| f(y)(\pi(y))^\theta - f(x)(\pi(y))^\theta \right| + \left| f(x)(\pi(y))^\theta - f(x)(\pi(x))^\theta \right| \\
    \nonumber
    \leq & \left|\pi(y)^\theta\right| \left|f(y) -f(x)\right| + \left|f(x)\right|\cdot \left|(\pi(y))^\theta - (\pi(x))^\theta\right|\\
    \leq & (L_f M_\pi^\theta + M_f L_\pi ) \eta_\theta. \label{1020_2}
\end{align}
From~\eqref{1020_2} and Assumption~\ref{assum_f_theta_lim} (3), it follows that $\int_{|y-x| \leq \eta_\theta}  \!\!\! \left|K_{\theta}(x-y)\right| \left| f(y)(\pi(y))^\theta - f(x)(\pi(x))^\theta \right|dy \leq M_K (L_f M_\pi^\theta + M_f L_\pi ) \eta_\theta$. Hence, when $\eta_\theta \leq \frac{\epsilon}{3M_K(L_f M_\pi + M_f L_\pi)}$, 
\be\label{1020_condi2}
\int_{|y-x| \leq \eta_\theta}\!\!\!  \!\!\!\left|K_{\theta}(x-y)\right| \left| f(y)(\pi(y))^\theta \!- \!f(x)(\pi(x))^\theta \right|dy \leq \! \frac{\epsilon}{3}
\ee
We now consider the third term in~\eqref{1020_1}. The assumption for $\theta$ in Assumption~\ref{assum_f_theta_lim} (4) yields that $(1-\epsilon/(3M))^{1/\theta} \leq m_\pi$ and $(1+\epsilon/(3M))^{1/\epsilon} \geq M_\pi$, that is 
\be\label{eqq}
(1-\epsilon/(3M))^{1/\theta} < \pi(x) < (1+\epsilon/(3M))^{1/\theta},
\ee
given $|\pi(x)| \in [m_\pi, M_\pi], \forall x \in \textnormal{supp} f$. Equivalently, rewriting~\eqref{eqq} as 
$\left|1-(\pi(x))^\theta\right| < \frac{\epsilon}{3M}$.
Therefore, Assumption~\ref{assum_f_theta_lim} (4) ensures
\be\label{1020_condi3}
\left|f(x)\right| \left|1-(\pi(x))^\theta \right| \leq M \left|1-(\pi(x))^\theta \right| < \frac{\epsilon}{3}.
\ee
Therefore,~\eqref{1020_condi1}~\eqref{1020_condi2}~\eqref{1020_condi3} indicates~\eqref{1020_1} is upper bounded by $\epsilon$, that is,
$\left|L_{K_\theta}f - f\right|_\infty \leq \epsilon$.
This completes the proof. 
\end{proof}

In this following, we provide the proof for Theorem~\ref{theorem_error_bound_KRR}. 

\begin{proof}
    From Proposition~\ref{prop_f_theta_in_range} and Proposition~\ref{prop_f_theta_lim}, $\forall \epsilon > 0$, $\exists \theta_\epsilon \in (0, 1)$, such that $f_{\theta_\epsilon}$ defined in~\eqref{eq_def_f_theta} is such that $f_{\theta_\epsilon} \in \operatorname{Range}(L_K)$ and
    $\lim_{\theta \to 0} \|f_{\theta_\epsilon} - f \|_{\infty} = \epsilon$. From Theorem~\ref{theorem_error_bound}, 
    \be
    \|f_{\theta_\epsilon} - \hat{f}_{\theta_\epsilon}^{\textnormal{KRR} }\|_\infty \leq 4 \log\!\left(\frac{2}{\delta}\right)
(3 \kappa M)^{\frac{1}{3}} \| L^{-1}_{K} L_{K_{\theta_\epsilon}}f \|_{L_\rho^2} m^{-\frac{1}{6}},
    \ee
    where $\| L^{-1}_{K} L_{K_{\theta_\epsilon}}f \|_{L_\rho^2} \leq c_{K} c_{f, s} \|l_{\pi, s, \theta }\|_{L^2}, l_{\pi, s, \theta} \in L^2(\mathbb{R}^d)$ from Proposition~\ref{prop_f_theta_in_range}. Therefore,  
    \be
    \|f_{\theta_\epsilon} - \hat{f}_{\theta_\epsilon}^{\textnormal{KRR} }\|_\infty \leq 4 \log\!\left(\!\frac{2}{\delta}\!\right)\!
(3 \kappa M)^{\frac{1}{3}} c_{K} c_{f, s} \|l_{\pi, s, \theta }\|_{L^2}  m^{-\frac{1}{6}}.
    \ee
    The error in~\eqref{eq_total_error} is as given in~\eqref{eq_KRR_bound}.
This completes the proof. 
\end{proof}

\subsection{Generalization Error}\label{sec_error2}
The probabilistic guarantee on the generalization error of OC-SVM was given in~\cite{OCSVM_generalization_error}.
\begin{theorem}
    Suppose that the OC-SVM gives an optimal solution $(M^*, \rho^*)$ with $\rho^* > 0$ for all $1 \leq i \leq m$. Then for any $\delta \in (0, 1)$ and $\epsilon > 0 $, with probability $1-\delta$ (over the choice of samples size $m$), we have
    \be
    \nonumber
    \mathbb{P}[\<M, \Gamma\> < \rho^* - \epsilon] \leq \frac{2}{m}\left(\left\lceil \log \tau(\|N\|_F, \epsilon) \right\rceil + \frac{2}{\delta}
\right),
    \ee
    where $\tau(\varrho, \epsilon, 2m)$ is the covering number of the model space $\{\<N, \cdot\>: \|N\|_F \leq \varrho\}$ by balls of radius of $\epsilon$  under the metric $d(N_1, N_2) = \textnormal{sup}_{\Gamma_1, \cdots, \Gamma_{2m}} \max_{i = 1, \cdots, 2m} |\<N_1 - N_2, \Gamma_i\>|$.
\end{theorem}
In~\cite{locicero2025issues}, the conclusion was translated explicitly
in terms of the dissipativity learning. Specifically,
by the definition of covering number, it could be found that
\[
\mathbb{P}\!\left[ \langle M, \Gamma \rangle < \rho^* \!-\! \epsilon \right]
\le
\frac{1}{m} \,
\mathcal{O}\!\left(
    \log \frac{1}{\delta}
    \! + \! \epsilon^{\frac{1}{2}} \log m
    \! + \! \frac{1}{\epsilon^2} \log \frac{1}{\epsilon}
\right)\!.
\]

\section{Numerical Results}\label{sec_simulation}
For the estimation of the dissipativity parameters $Q$ and $P$, we suppose that $m$ data points $(y^i(\cdot), u^i(\cdot))_{i = 1}^m$ and $(x^i(\cdot))_{i=1}^m$ are sampled, from a distribution of random input excitations such that the magnitude of the input satisfies Gaussian distribution with zero mean with a fixed duration for each magnitude. 
We consider a four-tray binary distillation column where the states are 
$\xv(t)=\begin{bmatrix}x_1(t) & x_2(t) & x_3(t) & x_4(t)\end{bmatrix}^{\!\top}$ with nonlinear dynamics as follows:
\begin{align}  
\nonumber
\dot{x}_1  =& -\frac{V}{H_1}x_1 
+ \frac{V\beta\,x_2}{H_1\!\left(1+(\beta-1)x_2\right)}, \\
\nonumber
\dot{x}_2  =& 
\frac{L}{H_2}(x_1-x_2)
+ \frac{V}{H_2}\!\left(
\frac{\beta\,x_3}{1+(\beta-1)x_3}
- \frac{\beta\,x_2}{1+(\beta-1)x_2}
\right), \\
\nonumber
\dot{x}_3  =&
\!- \frac{F}{H_3}x_3
\!+\! \frac{L}{H_3}(x_2\!-\!x_3)
\!+ \!\frac{V}{H_3} 
\frac{\beta\,x_4}{1+(\beta-1)x_4} - \frac{V}{H_3} \!\frac{\beta\,x_3}{1+(\beta-1)x_3} + \frac{F}{H_3}z_f\!\left(1+\operatorname{sat}(u(t))\right), \\
\nonumber
\dot{x}_4  =&
\frac{F+L}{H_4}(x_3-x_4)
+ \frac{V}{H_4}\!\left(
x_4 - \frac{\beta\,x_4}{1+(\beta-1)x_4}
\right), 
\end{align}
where 
$\operatorname{sat}(u) \eqdef \max\{-0.5,\min\{0.5,u(t)\}\}$.
The measured outputs are the top and bottom compositions
$y(t)=\begin{bmatrix} x_1(t) &x_4(t) \end{bmatrix}^{\!\top}$.
The parameters are listed in Table~\ref{table}. 
\begin{table}
\centering
\caption{Parameters of the distillation column.}
\label{table}
\begin{tabular}{lcl}
\toprule
Parameter & Symbol & Value \\ 
\midrule
Vapor flowrate & $V$ & 6.05 (mol/min) \\ 
Reflux flowrate & $L$ & 4.79 (mol/min) \\ 
Feed flowrate & $F$ & 1.70 (mol/min) \\ 
Relative volatility constant & $\beta$ & 1.60 \\ 
Nominal feed composition & $z_f$ & 0.56 \\ 
Liquid holdup & $H_1$ & 5.25 mol\\ 
& $H_2$ & 0.53 mol \\ 
& $H_3$ & 0.53 mol \\ 
& $H_4$ & 5.25 mol \\ 
\bottomrule
\end{tabular}
\end{table}
 
\subsection{Numerical stability of nonparametric and parametric dissipativity learning}
In the nonparametric RKHS formulation, we consistently observe that as \(\lambda\) increases, 
the objective value and the Hilbert--Schmidt norms \(\|\Pi\|_{\mathrm{HS}}\) and 
\(\|P\|_{\mathrm{HS}}\) grow monotonically, as shown in the top panel of 
Fig.~\ref{fig:objnorm_comparison}. This behavior aligns with the structure of 
\eqref{eq_def_nonparam_obj}, where a larger \(\lambda\) imposes a stronger penalty on the term 
\(\lambda \alpha\). The increased penalty drives the optimizer to reduce \(\alpha\), which in turn 
tightens the gain-like constraint 
\(-\Pi - \Upsilon_0 + \alpha \Upsilon_1 \succeq 0\) 
and enforces a stricter gain bound. As a result, the feasible set for \(\Pi\) contracts, and the 
optimizer compensates by increasing the magnitude of \(\Pi\), leading to a systematic growth in 
its Hilbert--Schmidt norm. 

Similarly, the norm of \(P\) also increases systematically with \(\lambda\). 
This trend is coupled with the variation of \(\Pi\) through the sample-wise dissipativity 
constraints 
\(\langle \Pi, \Gamma_t \rangle + \langle P, \Delta_t \rangle \ge \rho\). 
When \(\alpha\) decreases and \(\Pi\) becomes more negative, the term 
\(\langle \Pi, \Gamma_t \rangle\) generally decreases for most samples. 
To preserve the feasibility of the dissipativity inequalities, the optimizer compensates by 
increasing the contribution from \(P\), so that 
\(\langle P, \Delta_t \rangle\) grows accordingly. 
This adjustment results in a consistent increase in the Hilbert--Schmidt norm 
\(\|P\|_{\mathrm{HS}}\) as \(\lambda\) increases. 
In other words, both \(\|\Pi\|_{\mathrm{HS}}\) and \(\|P\|_{\mathrm{HS}}\) must increase jointly to 
satisfy the tightened dissipativity constraints under a stricter gain bound. 
Consequently, the overall objective value and the operator norms exhibit a monotone 
nondecreasing trend with respect to \(\lambda\), as illustrated in 
Fig.~\ref{fig:objnorm_comparison}. 
Furthermore, as expected, increasing \(\lambda\) enforces a stricter gain bound, reflected by the 
smaller optimal values of \(\alpha\) shown in the top panel of 
Fig.~\ref{fig:rhoalpha_comparison}.

\begin{figure}[htbp]
	\centering
	\includegraphics[width=0.8\textwidth]{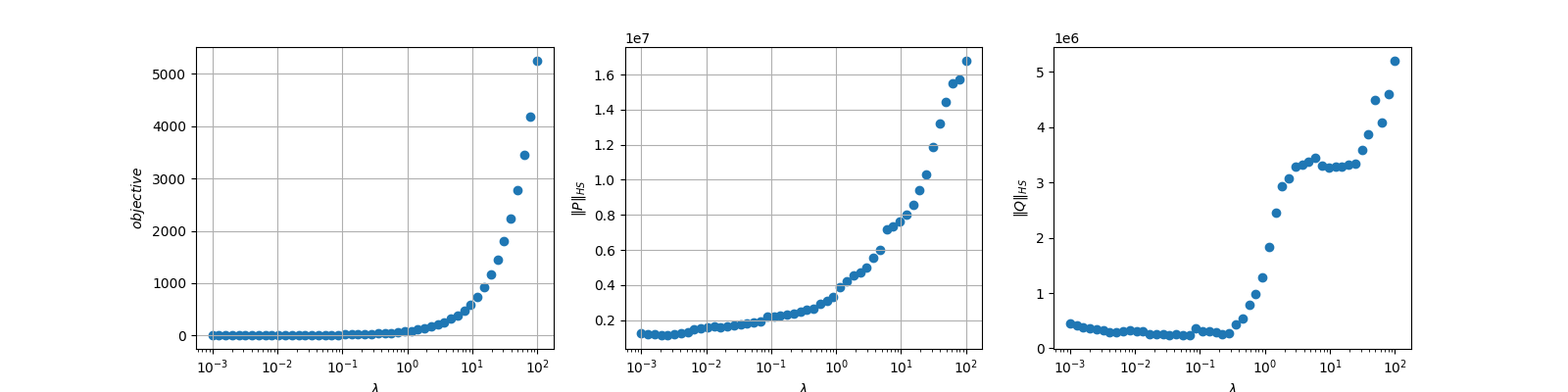}
	\vspace{1em} 
	\includegraphics[width=0.8\textwidth]{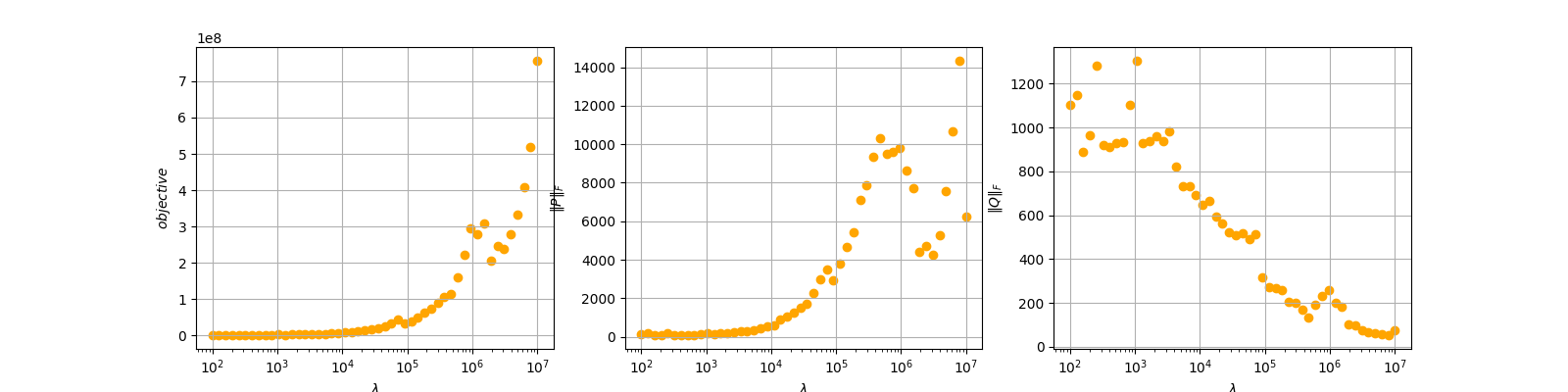}
	\caption{Comparison of the objectives of nonparametric dissipativity learning in RKHS (top) and parametric learning (bottom).}
	\label{fig:objnorm_comparison}
\end{figure}
 
\begin{figure}[htbp]
	\centering
	\includegraphics[width=0.8\textwidth]{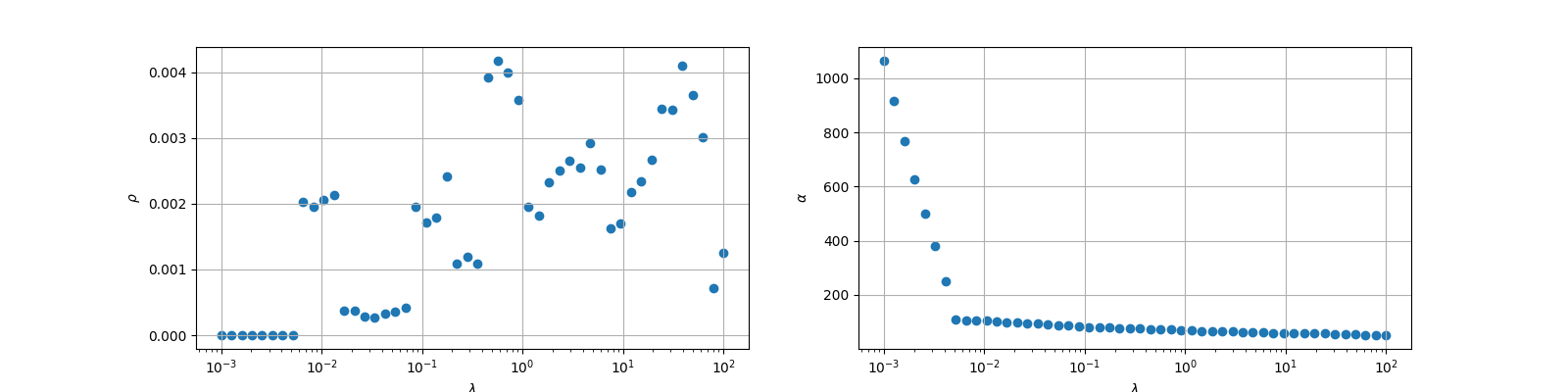}
	\vspace{1em} 
	\includegraphics[width=0.8\textwidth]{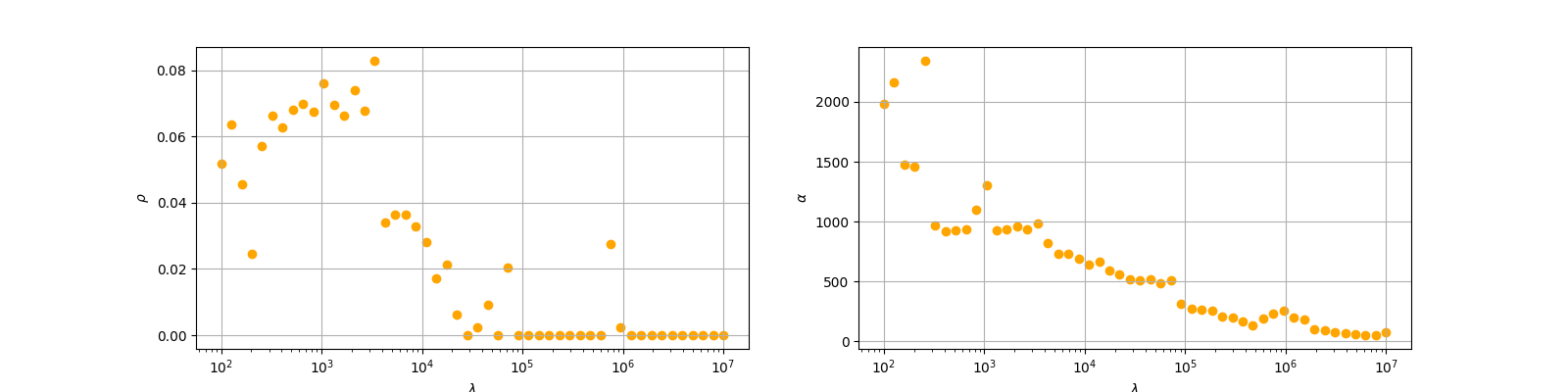}
	\caption{Comparison of the learned dissipativity margin $\rho$ and $\alpha$ from nonparametric dissipativity learning in RKHS (top) and parametric learning (bottom), with respect to weighting parameter $\lambda$.}
	\label{fig:rhoalpha_comparison}
\end{figure}
Similarly, in the parametric learning setting, the same monotonicity property holds. The objective 
function is convex and nondecreasing with respect to \(\lambda\). 
As in the nonparametric case, increasing \(\lambda\) penalizes large values of \(\alpha\), thus enforcing a stricter dissipativity gain bound.
As \(\lambda\) increases, it leads to a reduction in 
\(\alpha\) and thereby small feasible region of \(Q\). This is enforced by the constraint in~\eqref{eq_obj_parametric}, i.e. $Q \prec  \begin{bmatrix}
\alpha  & 0 \\
0 & - I
\end{bmatrix}$. To maintain the feasibility, the matrix $Q$ is pushed away from the origin in the matrix space, leading to an increase in its Frobenius norm. This reduces the term $\langle Q, \Gamma_t \rangle $. To preserve feasibility and maintain the dissipativity margin \(\rho\) in \(\langle Q, \Gamma_t \rangle + \langle P, \Delta_t \rangle \ge \rho\), the optimizer compensates 
by increasing \(\langle P, \Delta_t \rangle\) across samples. 
Achieving this compensation requires \(P\) to have larger eigenvalues in magnitude, 
thereby increasing its Frobenius norm. 
In other words, as the gain constraint on \(Q\) becomes stricter, \(P\) must amplify its effect to 
balance the tightened inequalities, leading to a systematic growth in \(\|P\|_{F}\) as 
\(\lambda\) increases too. This analysis is shown numerically in the bottom panel of Fig.~\ref{fig:objnorm_comparison}. It is worth noting that the nonparametric learning formulation exhibits better numerical stability in comparison with its parametric counterpart.

The verification for the dissipativity learning is shown in Fig.~\ref{fig:verification_parametric}, which evaluates the learned 
dissipativity margin on testing data. For each weighting parameter
\(\lambda\), the sample-wise dissipativity margin was computed using the learned matrices 
\(P\) and \(Q\), normalized by \(\|P\|_F + \|Q\|_F\).
The mean and variance of these normalized margins are shown as error bars, together with the 
learned dissipativity margin \(\rho\) obtained from the training data.
The results confirm that the dissipativity inequality remains valid on the verification data for 
a wide range of \(\lambda\), as the test margin remains positive and closely follow the 
trend as the learned margin on training data. As \(\lambda\) increases, the test margin decreases slightly, 
reflecting the expected trade-off between conservativeness and data fit. A larger \(\lambda\) 
enforces a stricter gain bound and increases the magnitudes of \(P\) and \(Q\), 
leading to a more conservative but stable certificate. The relatively small standard deviations 
indicate that the learned quadratic storage and supply matrices generalize well to unseen data. 
Overall, the verification demonstrates that the proposed approach yields a 
consistent dissipativity characterization that is robust with respect to both the regularization 
parameter \(\lambda\) and unseen data.

\begin{figure}[htbp]
    \centering
    \includegraphics[width=0.7\textwidth]{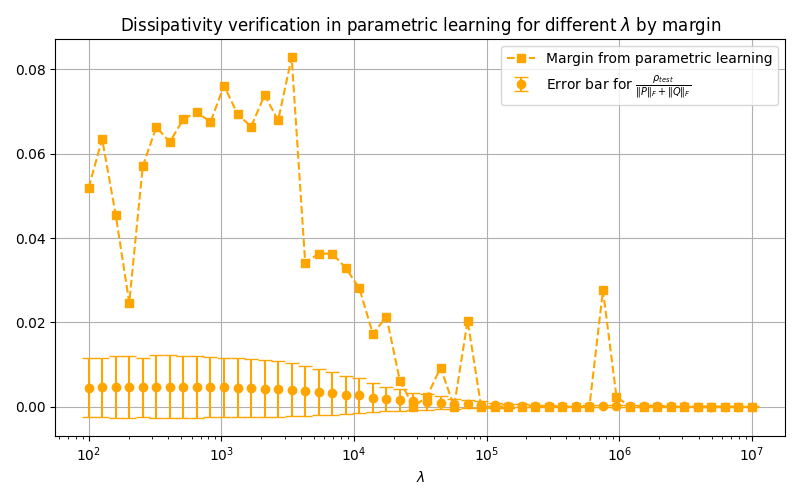}
    \caption{Dissipativity verification for parametric learning.}
    \label{fig:verification_parametric}
\end{figure}

\section{Conclusion}\label{sec_conclusion}
This paper introduced a nonparametric framework for dissipativity learning in RKHS by casting the problem in an OC-SVM formulation. By leveraging representer properties, we reduced the infinite\textendash dimensional operator problem to a finite\textendash dimensional convex problem expressed solely through kernel Gram matrices, thereby retaining expressiveness while ensuring tractability. We further provided a principled error analysis that separates (i) the kernel regression approximation error used to construct reference operators and (ii) the OC\textendash SVM generalization error on the empirical margin, yielding population\textendash level guarantees on the learned dissipativity certificates and the implied $\mathcal{L}_2$ gain bound.

Numerical experiments on a nonlinear distillation column demonstrate that the RKHS formulation reliably identifies dissipative behavior directly from input--output trajectories, generalizes to unseen data, and exhibits favorable numerical stability relative to a parametric baseline. As the tradeoff parameter $\lambda$ increases, we observed a monotone, nondecreasing trend in the objective and operator norms, consistent with the tightening of the gain\textendash like constraint. Concurrently, the verified margins indicates the robustness of the approach.

\bibliographystyle{IEEEtran}
\bibliography{bib}

\end{document}